\documentclass[english]{llncs}
\usepackage[T1]{fontenc}
\usepackage[latin9]{inputenc}
\usepackage{amsmath}
\usepackage{todonotes}
\usepackage{fancyhdr}
\usepackage{tabularx,caption}
\usepackage[ruled,linesnumbered]{algorithm2e}
\usepackage{babel}
\usepackage{hyperref}
\usepackage{multirow}
\usepackage{wrapfig}
 \usepackage{subcaption}



\newcommand{\ie}{i.\,e.\xspace}
\newcommand{\eg}{e.\,g.\xspace}

\newcommand{\vd}{$\mathit{VD}$\xspace}

\newcommand{\rk}{\textsf{RK}\xspace}
\newcommand{\ia}{\textsf{IA}\xspace}
\newcommand{\iaw}{\textsf{IAW}\xspace}
\newcommand{\da}{\textsf{DA}\xspace}
\newcommand{\daw}{\textsf{DAW}\xspace}
\newcommand{\vda}{$\tilde{\mathit{VD}}$\xspace}
\newcommand{\upvd}{\textsf{updateApprVD-W}\xspace}
\newcommand{\sssp}{\textsf{updateSSSP-W}\xspace}
\newcommand{\upvdu}{\textsf{updateApprVD-U}\xspace}
\newcommand{\ssspu}{\textsf{updateSSSP-U}\xspace}

\include{header}

\begin{document}

\author{Elisabetta Bergamini \and Henning Meyerhenke}
\institute{Institute of Theoretical Informatics \\ Karlsruhe Institute of Technology (KIT), Germany \\
  Email: \email{\{elisabetta.bergamini, meyerhenke\}\,@\,kit.edu}
}

\title{Fully-dynamic Approximation \\ of Betweenness Centrality}

\date{}

\maketitle

\begin{abstract}
Betweenness is a well-known centrality measure that ranks the nodes of a network according to their participation in shortest paths.
Since an exact computation is prohibitive in large networks, several approximation algorithms have been proposed. Besides that, recent years
have seen the publication of dynamic algorithms for efficient recomputation of betweenness in evolving networks. In previous work  we
proposed the first semi-dynamic algorithms that recompute an \emph{approximation} of betweenness in connected graphs after batches of edge insertions. 

In this paper we propose the first fully-dynamic approximation algorithms (for weighted and unweighted undirected graphs that need not to be connected) with a provable guarantee on the maximum approximation error.
The transfer to fully-dynamic and disconnected graphs implies additional algorithmic problems that could be of independent interest. In particular, we propose a new upper bound on the vertex diameter for weighted undirected graphs. For both weighted and unweighted graphs, we also propose the first fully-dynamic algorithms that keep track of this upper bound.
In addition, we extend our former algorithm for semi-dynamic BFS to batches of both edge insertions and deletions. 

Using approximation, our algorithms are the first to make in-memory computation of betweenness in fully-dynamic networks with millions of edges feasible. Our experiments show that they can achieve substantial speedups compared to recomputation, up to 
several orders of magnitude.\\[0.25ex]
\end{abstract}


\section{Introduction}
\label{sec:intro}
The identification of the most central nodes of a network is a fundamental problem in network analysis. \emph{Betweenness
centrality} (BC) is a well-known index that ranks the importance of nodes according to their participation in \textit{shortest paths}. Intuitively, a node has 
high BC when it lies on many shortest paths between pairs of other nodes.
Formally, BC of a node $v$ is defined as $c_B(v) = \frac{1}{n(n-1)} \sum_{s \neq v \neq t} \frac{\sigma_{st}(v)}{\sigma_{st}}$, where $n$ is the number of nodes,
$\sigma_{st}$ is the number of shortest paths between two nodes $s$ and $t$ and $\sigma_{st}(v)$ is the number of these paths that go through node $v$.
Since it depends on \textit{all} shortest paths, the exact computation of BC is expensive: the best known 
algorithm~\cite{Brandes01betweennessCentrality} is quadratic in the number of nodes for sparse networks and cubic for dense networks, prohibitive for networks with hundreds of thousands of nodes. Many graphs of interest, however, such as web 
graphs or social networks, have millions or even billions of nodes and edges. 
For this reason, approximation algorithms~\cite{DBLP:journals/ijbc/BrandesP07,DBLP:conf/alenex/GeisbergerSS08,DBLP:conf/waw/BaderKMM07} must be used in practice. In addition, many
large graphs of interest evolve continuously, making the efficient recomputation of BC a necessity.
In a previous work, we proposed the first two approximation algorithms~\cite{DBLP:conf/alenex/BergaminiMS15} (\ia for unweighted and \iaw for weighted graphs) 
that can efficiently recompute the approximate BC scores after batches of edge insertions or weight decreases. 
\ia and \iaw are the only semi-dynamic algorithms that can actually be applied to large networks. 
The algorithms build on \rk~\cite{DBLP:conf/wsdm/RiondatoK14}, a static algorithm with a theoretical guarantee on the quality of the approximation, and inherit this guarantee from \rk. 
However, \ia and \iaw target a relatively restricted configuration: only connected graphs and edge insertions/weight decreases.
 \vspace{-1ex}
\paragraph{Our contributions.}
In this paper we present the first fully-dynamic algorithms (handling edge insertions, deletions and arbitrary weight updates) for BC approximation in weighted and unweighted undirected graphs. Our algorithms extend the semi-dynamic ones we presented in~\cite{DBLP:conf/alenex/BergaminiMS15}, while keeping the theoretical guarantee on the maximum approximation error.
The transfer to fully-dynamic and disconnected graphs implies several additional problems compared to the restricted case we considered previously~\cite{DBLP:conf/alenex/BergaminiMS15}.
Consequently, we present the following intermediate results, all of which could be of independent interest.
(i) We propose a new upper bound on the vertex diameter \vd (\ie number of nodes in the shortest path(s) with the maximum number of nodes) for weighted undirected graphs. This can improve significantly the one used in the \rk algorithm~\cite{DBLP:conf/wsdm/RiondatoK14} if the network's weights vary in relatively small ranges (from the size of the largest connected component to at most twice the vertex diameter times the ratio between the maximum and the minimum edge weights).
(ii) For both weighted and unweighted graphs, we present the first fully-dynamic algorithm for updating an approximation of \vd, which is equivalent to the diameter in unweighted graphs.
(iii) We extend our previous semi-dynamic BFS algorithm~\cite{DBLP:conf/alenex/BergaminiMS15} to batches of both edge insertions and deletions.
In our experiments, we compare our algorithms to recomputation with \rk on both synthetic and real dynamic networks. Our results show that our algorithms can achieve substantial speedups, often several orders of magnitude on single-edge updates and are always faster than recomputation on batches of more than 1000 edges. 
\section{Related work}
\label{sec:related_work}
\subsection{Overview of algorithms for computing BC}
The best static exact algorithm for BC (\textsf{BA}) is due to Brandes~\cite{Brandes01betweennessCentrality} and requires $\Theta(nm)$ operations for unweighted graphs
and $\Theta(nm+n^{2}\log n)$ for graphs with positive edge weights. The algorithm computes a single-source shortest path (SSSP) search from every node $s$ in the graph and adds to the BC score of each node $v \neq s$ the fraction of shortest paths that go through $v$. Several static approximation algorithms have been proposed that compute an SSSP search from a set of randomly chosen nodes and extrapolate the BC scores of the other nodes~\cite{DBLP:journals/ijbc/BrandesP07,DBLP:conf/alenex/GeisbergerSS08,DBLP:conf/waw/BaderKMM07}. The static approximation algorithm by Riondato and Kornaropoulos (\rk)~\cite{DBLP:conf/wsdm/RiondatoK14} samples a set of shortest paths and adds a contribution to each node in the sampled paths. This approach allows a theoretical guarantee on the quality of the approximation and will be described in Section~\ref{sec:rk}.
Recent years have seen the publication of a few dynamic exact algorithms~\cite{DBLP:conf/www/LeeLPCC12,DBLP:conf/socialcom/GreenMB12,kourtellis2014scalable,DBLP:journals/snam/KasCC14,DBLP:conf/mfcs/NasrePR14,DBLP:conf/waw/GoelSIS13}. Most of them
store the previously calculated BC values and additional information, like the distance
of each node from every source, and try to limit the recomputation to the nodes whose
BC has actually been affected. All the dynamic algorithms perform better than
recomputation on certain inputs. Yet, none of them
is in general better than \textsf{BA}. In fact, they
all require updating an all-pairs shortest paths (APSP) search, for which no algorithm has an improved worst-case complexity compared to the best static algorithm~\cite{DBLP:journals/algorithmica/RodittyZ11}. 
Also, the scalability of the dynamic exact BC algorithms is strongly compromised by their memory requirement of $\Omega (n^2)$. To overcome these problems, we presented two algorithms that efficiently recompute an approximation of the BC scores instead of their exact values~\cite{DBLP:conf/alenex/BergaminiMS15}. 
The algorithms have shown significantly high speedups compared to recomputation with \rk and a good scalability, but they are limited to connected graphs and batches of edge insertions/weight decreases (see Section~\ref{sec:bms}).

\subsection{\rk algorithm}
\label{sec:rk}
The static approximation algorithm \rk~\cite{DBLP:conf/wsdm/RiondatoK14} is
the foundation for the incremental approach we presented in~\cite{DBLP:conf/alenex/BergaminiMS15} and our new fully-dynamic approach. \rk samples a set $S =\{p_{(1)},...,p_{(r)}\}$ of $r$ shortest paths between randomly-chosen source-target pairs $(s, t)$.
Then, \rk computes the approximated
betweenness $\tilde{c}_B(v)$ of a node $v$ as the
fraction of sampled paths $p_{(k)}\in S$ that go through $v$, by adding $\frac{1}{r}$ to $v$'s score for each of these paths.
In each of the $r$ iterations, the probability of a shortest path $p_{st}$ to be sampled is $\pi_{G}(p_{st})=\frac{1}{n(n-1)}\cdot\frac{1}{\sigma_{st}}$. The number $r$ of samples required to approximate the BC scores with the given error guarantee is $r=\frac{c}{\epsilon^{2}}\left(\lfloor\log_{2}\left(\mathit{VD}-2\right)\rfloor+1+\ln\frac{1}{\delta}\right)$, where $\epsilon$ and $\delta$ are constants in $(0,1)$ and $c \approx 0.5$.
Then, if $r$ shortest paths
are sampled according to $\pi_{G}$, with probability at least $1-\delta$ the approximations
$\tilde{c}_B(v)$ are within $\epsilon$
from their exact value: $ \Pr(\exists v\in V\: s.t.\:|c_{B}(v)-\tilde{c}_B(v)|>\epsilon)<\delta. $
To sample the shortest paths according to $\pi_{G}$, \rk first chooses
a source-target node pair $(s,t)$ uniformly at random and performs a shortest-path search (Dijkstra or BFS) from $s$ to $t$, keeping also track of the number $\sigma_{sv}$
of shortest paths between $s$ and $v$ and of the list of predecessors
$P_{s}(v)$ (\ie the nodes that immediately precede $v$ in the shortest paths between $s$ and $v$) for any node $v$ between $s$ and $t$. Then one shortest path is selected: 
starting from $t$, a predecessor $z\in P_{s}(t)$
is selected with probability $\sigma_{sz}/\sum_{w\in P_{s}(t)}\sigma_{sw}=\sigma_{sz}/\sigma_{st}$.
The sampling is repeated iteratively
until node $s$ is reached.
\paragraph{Approximating the vertex diameter.}
\rk uses two upper bounds on \vd that can be both computed in $O(n+m)$. For unweighted undirected graphs, it samples a source node $s_i$ for each connected component of $G$, computes a BFS from each $s_i$ and sums the two shortest paths with maximum length starting in $s_i$. The \vd approximation is the maximum of these sums over all components. For weighted graphs, \rk approximates \vd with the size of the largest connected component, which can be a significant overestimation for complex networks, possibly of orders of magnitude. In this paper, we present a new approximation for weighted graphs, described in Section~\ref{sec:new_vd_approx}.
\subsection{\textsf{IA} and \textsf{IAW} algorithms}
\label{sec:bms}
\ia and \iaw are the incremental approximation algorithms (for unweighted and weighted graphs, respectively) that we presented previously~\cite{DBLP:conf/alenex/BergaminiMS15}. The algorithms are based on the observation that if only edge insertions are allowed and the graph is connected, \vd cannot increase, and therefore also the number $r$ of samples required by \rk for the theoretical guarantee. Instead of recomputing $r$ new shortest paths after a batch of edge insertions, \ia and \iaw \textit{replace} each old shortest path $p_{s,t}$ with a new shortest path between the same node pair $(s,t)$. In \iaw the paths are recomputed with a slightly-modified \textsf{T-SWSF}~\cite{DBLP:conf/wea/BauerW09}, whereas \ia uses a new semi-dynamic BFS algorithm.
The BC scores are updated by subtracting $1/r$ to the BC of the nodes in the old path and adding $1/r$ to the BC of nodes in the new shortest path.

\subsection {Batch dynamic SSSP algorithms}
\label{sssp_update}
Dynamic SSSP algorithms recompute distances from a source node after a single edge update or a batch of edge updates.
Algorithms for the batch problem have been published~\cite{Ramalingam92anincremental,Frigioni_semi-dynamicalgorithms,DBLP:conf/wea/BauerW09} and compared in experimental studies~\cite{DBLP:conf/wea/BauerW09,DBLP:conf/wea/DAndreaDFLP14}.
The experiments show that the tuned algorithm \textsf{T-SWSF} presented in~\cite{DBLP:conf/wea/BauerW09} performs well on many types of graphs and edge updates. For batches of only edge insertions in unweighted graphs, we developed an algorithm asymptotically faster than \textsf{T-SWSF}~\cite{DBLP:conf/alenex/BergaminiMS15}. The algorithm is in principle similar to \textsf{T-SWSF}, but has an improved complexity thanks to different data structures.

\section{New \vd approximation for weighted graphs}
\label{sec:new_vd_approx}
Let $G$ be an undirected graph. For simplicity, let $G$ be connected for now. If it is not, we compute an approximation for each connected component and take the maximum over all the approximations. Let $T \subseteq G$ be an SSSP tree from any source node $s \in V$. Let $p_{xy}$ denote a shortest path between $x$ and $y$ in $G$ and let $p_{xy}^T$ denote a shortest path between $x$ and $y$ in $T$. Let $|p_{xy}|$ be the number of nodes in $p_{xy}$ and $d(x,y)$ be the distance between $x$ and $y$ in $G$, and analogously for $|p_{xy}^T|$ and $d^T(x,y)$. Let $\overline{\omega}$ and  $\underline{\omega}$ be the maximum and minimum edge weights, respectively. Let $u$ and $v$ be the nodes with maximum distance from $s$, \ie $d(s, u)\geq d(s,v) \geq d(s, x)\  \forall x\in V, x\neq u $. 

We define the \vd approximation $\tilde{\mathit{VD}} := 1 + \frac{d(s,u)+d(s,v)}{\underline{\omega}}$. Then:
\begin{proposition}
\label{lem:vd2}
$\mathit{VD} \leq \tilde{\mathit{VD}} < 2 \cdot \frac{\overline{\omega}}{\underline{\omega}} \mathit{VD} $. (Proof in Section~\ref{sub:proof_vd2}, Appendix)
\end{proposition}
To obtain the upper bound \vda, we can simply compute an SSSP search from any node $s$, find the two nodes with maximum distance and perform the remaining calculations.
Notice that \vda extends the upper bound proposed for \rk~\cite{DBLP:conf/wsdm/RiondatoK14} for unweighted graphs: When the graph is unweighted and thus $\underline{\omega} = \overline{\omega}$, \vda becomes equal to the approximation used by \rk.
Complex networks are often characterized by a small diameter and in networks like coauthorship, friendship, communication networks, \vd and $ \frac{\overline{\omega}}{\underline{\omega}}$ can be several order of magnitude smaller than the size of the largest component. This translates into a substantially improved \vd approximation.

\section{New fully-dynamic algorithms}
\paragraph{Overview.}
We propose two fully-dynamic algorithms, one for unweighted (\da, dynamic approximation) and one for weighted (\daw, dynamic approximation weighted) graphs. Similarly to \textsf{IA} and \textsf{IAW}, our new fully-dynamic algorithms keep track of the old shortest paths and substitute them only when necessary. However, if $G$ is not connected or edge deletions occur, \vd can grow and a simple substitution of the paths is not sufficient anymore. 
Although many real-world networks exhibit a shrinking-diameter behavior~\cite{DBLP:conf/kdd/LeskovecKF05}, to ensure our theoretical guarantee, we need to keep track of \vda over time and sample new paths in case \vda increases.
The need for an efficient update of \vda augments significantly the difficulty of the fully-dynamic problem, as well as the necessity to recompute the SSSPs after batches of both edge insertions and deletions. 
The building block for the BC update are basically two: a fully-dynamic algorithm that updates distances and number of shortest paths from a certain source node (SSSP update) and an algorithm that keeps track of a \vd approximation for each connected component of $G$. The following paragraphs give an overview of such building blocks, which could be of independent interest. The last paragraph outlines the dynamic BC approximation algorithm. \textbf{Due to space constraints, a detailed description of the algorithms as well as the pseudocodes and the omitted proofs can be found in the Appendix}.
\paragraph{SSSP update in weighted graphs.}
Our SSSP update is based on \textsf{T-SWSF}~\cite{DBLP:conf/wea/BauerW09}, which recomputes distances from a source node $s$ after a batch $\beta$ of weight updates (or edge insertions/deletions). For our BC algorithm, we need two extensions of \textsf{T-SWSF}: an algorithm that also recomputes the number of shortest paths between $s$ and the other nodes (\sssp) and one that also updates a \vd approximation for the connected component of $s$ (\upvd). 
The \vd approximation is computed as described in Section~\ref{sec:new_vd_approx}. Thus, \upvd keeps track of the two maximum distances $d'$ and $d''$ from $s$ and the minimum edge weight $\underline{\omega}$.
We call \textit{affected nodes} the nodes whose distance (or also whose number of shortest paths, in \sssp) from $s$ has changed as a consequence of $\beta$. Basically, the idea is to put the set $A$ of affected nodes $w$ into a priority queue $Q$ with priority $p(w)$ equal to the candidate distance of $w$. When $w$ is extracted, if there is actually a path of length $p(w)$ from $s$ to $w$, the new distance of $w$ is set to $p(w)$, otherwise $w$ is reinserted into $Q$ with a higher candidate distance. In both cases, the affected neighbors of $w$ are inserted into $Q$. In \upvd, $d'$ and $d''$ are recomputed while updating the distances and $\underline{\omega}$ is updated while scanning $\beta$.  In \sssp, the number $\sigma_{sw}$ of shortest paths of $w$ is recomputed as the sum of the $\sigma_{sz}$ of the new predecessors $z$ of $w$. 

Let $|\beta |$ represent the cardinality of $\beta$ and let $||A ||$ represent the sum of the nodes in $A$ and of the edges that have at least one endpoint in $A$. Then, the following complexity derives from feeding $Q$ with the batch and inserting into/extracting from $Q$ the affected nodes and their neighbors.
\begin{lemma}
\label{thm:complexity}
The time required by \upvd (\sssp) to update the distances and \vda (the number of shortest paths) is $O(|\beta |\log |\beta | + ||A || \log ||A ||)$.
\end{lemma}
\paragraph{SSSP update in unweighted graphs.}
For unweighted graphs, we basically replace the priority queue $Q$ of \upvd and \sssp with a list of queues, as the one we used in~\cite{DBLP:conf/alenex/BergaminiMS15} for the incremental BFS. Each queue represents a \emph{level} from 0 (which only the source belongs to) to the maximum distance $d'$. The levels replace the priorities and also in this case represent the candidate distances for the nodes. In order not to visit a node multiple times, we use colors to distinguish the unvisited nodes from the visited ones. The replacement of the priority queue with the list of queues decreases the complexity of the SSSP update algorithms for unweighted graphs, that we call \upvdu and \ssspu, in analogy with the ones for weighted graphs.
\begin{lemma}
\label{thm:complexity2}
The time required by \upvdu (\ssspu) to update the distances and \vda (the number of shortest paths) is $O(|\beta |+ ||A ||+d_{\max})$, where $d_{\max}$ is the maximum distance from $s$ reached during the update.
\end{lemma}
\paragraph{Fully-dynamic \vd approximation.}
The algorithm keeps track of a \vd approximation for the whole graph $G$, \ie for each connected component of $G$. It is composed of two phases. In the initialization, we compute an SSSP from a source node $s_i$ for each connected component $C_i$. During the SSSP search from $s_i$, we also compute a \vd approximation $\tilde{\mathit{VD}_i}$ for $C_i$, as described in Sections~\ref{sec:rk} and~\ref{sec:new_vd_approx}. In the update, we recompute the SSSPs and the \vd approximations with \upvd (or \upvdu). Since components might split or merge, we might need to compute new approximations, in addition to update the old ones. To do this, for each node, we keep track of the number of times it has been visited. This way we discard source nodes that have already been visited and compute a new approximation for components that have become unvisited. The complexity of the update of the \vd approximation derives from the \vda update in the single components, using \upvd and \upvdu.
\begin{theorem}
\label{thm:complexity_vd}
The time required to update the \vd approximation is $O(n_c\cdot |\beta |\log |\beta |+\sum_{i=1}^{n_c}||A^{(i)}|| \log ||A^{(i)}||)$ in weighted graphs and $O(n_c\cdot |\beta |+\sum_{i=1}^{n_c}||A^{(i)}||+d^{(i)}_{max})$ in unweighted graphs, where $n_c$ is the number of components in $G$ before the update and $A^{(i)}$ is the sum of affected nodes in $C_i$ and their incident edges.
\end{theorem}
\paragraph{Dynamic BC approximation.}
Let $G$ be an undirected graph with $n_c$ connected components. Now that we have defined our building blocks, we can outline a fully-dynamic BC algorithm: we use the fully dynamic \vd approximation to recompute \vda after a batch, we update the $r$ sampled paths with \textsf{updateSSSP} and, if \vda (and therefore $r$) increases, we sample new paths. However, since \textsf{updateSSSP} and \textsf{updateApprVD} share most of the operations, we can ``merge'' them and update at the same time the shortest paths from a source node $s$ and the \vd approximation for the component of $s$. We call such hybrid function \textsf{updateSSSPVD}. Instead of storing and updating $n_c$ SSSPs for the \vd approximation and $r$ SSSPs for the BC scores, we 
recompute a \vd approximation for each of the $r$ samples while recomputing the shortest paths with \textsf{updateSSSPVD}. This way we do not need to compute an additional SSSP for the components covered by $r$ sampled paths (\ie in which the paths lie), saving time and memory. Only for components that are not covered by any of them (if they exist), we compute and store a separate \vd approximation. We refer to such components as $R'$ (and to $|R'|$ as $r'$). 

\begin{algorithm}[h]
\begin{small}
\LinesNumbered
\SetKwData{B}{$\tilde{c}_B$}
\SetKwFunction{dynamicSSSP}{updateSSSP}
\SetKwFunction{updateApproxVD}{updateApprVD}
\SetKwFunction{updateApproxVDP}{updateSSSPVD}
\SetKwFunction{computeApproxVD}{initApprVD}
\SetKwFunction{replacePath}{replacePath}
\SetKwFunction{sampleNewPaths}{sampleNewPaths}
\SetKwFunction{applyBatch}{applyBatch}
\applyBatch{$G, \beta$}\; \label{bc:batch}
\For{$i \leftarrow 1$ \KwTo $r$}
{\label{bc:update1}
	$\tilde{\mathit{VD}}_i \leftarrow$ \updateApproxVDP{$s_i, \beta$}\;\label{bc:dynsssp}
	\replacePath{$s_i, t_i$} \tcc*{update of BC scores}
}\label{bc:update2}
\ForEach{$C_i \in R'$}
{\label{bc:rprime1}
	$\tilde{\mathit{VD}}_i \leftarrow$ \updateApproxVD{$C_i, \beta$}\;
}\label{bc:rprime2}
\ForEach{unvisited $C_j$}
{\label{bc:unvisited1}
	add $C_j$ to $R'$\;
	$\tilde{\mathit{VD}}_j \leftarrow$ \computeApproxVD{$C_j$}\;
}\label{bc:unvisited2}
\vda $\leftarrow \max_{C_i \in R \cup R'} \tilde{\mathit{VD}}_i$\; \label{bc:recompute1}
$r_{\text{new}} \leftarrow (c/\epsilon^2) (\lfloor \log_2(\tilde{\mathit{VD}}-2)\rfloor +\ln(1/\delta))$\; \label{bc:recompute2}
\If{$r_{\text{new}}>r$}
{\label{bc:norm1}
	\sampleNewPaths{} \tcc*{update of BC scores}
	\ForEach{$v \in V$}
	{
		$\B(v) \leftarrow \B(v) \cdot r/r_{\text{new}}$ \label{bc:norm1} \tcc*{renormalization of BC scores}
	}
$r \leftarrow r_{\text{new}}$\;
}\label{bc:norm2}
\Return{$\{(v,\B(v)):\: v\in V\}$}
\end{small}
\caption{BC update after a batch $\beta$ of edge updates}
\label{bc_overview}
\end{algorithm}
%
The high-level description of the update after a batch $\beta$ is shown as Algorithm~\ref{bc_overview}. 
After changing the graph according to $\beta$ (Line~\ref{bc:batch}), we recompute the previous $r$ samples and the \vd approximations for their components (Lines~\ref{bc:update1}~-~\ref{bc:update2}). Then, similarly to \ia and \iaw, we update the BC scores of the nodes in the old and in the new shortest paths. 
Thus, we update a \vd approximation for the components in $R'$ (Lines~\ref{bc:rprime1}~-~\ref{bc:rprime2}) and compute a new approximation for new components that have formed applying the batch (Lines~\ref{bc:unvisited1}~-~\ref{bc:unvisited2}). Then, we use the results to update the number of samples (Lines~\ref{bc:recompute1}~-~\ref{bc:recompute2}). If necessary, we sample additional paths and normalize the BC scores (Lines~\ref{bc:norm1}~-~\ref{bc:norm2}). The difference between \da and \daw is the way the SSSPs and the \vd approximation are updated: in \da we use \upvdu and in \daw \upvd. Differently from \rk and our previous algorithms \ia and \iaw, in \da and \daw we scan the neighbors every time we need the predecessors instead of storing them. This allows us to use $\Theta(n)$ memory per sample (\ie, $\Theta((r+r')n)$ in total) instead of $\Theta(m)$ per sample, while our experiments show that the running time is hardly influenced. The number of samples depends on $\epsilon$, so in theory this can be as large as $|V|$. However, the experiments conducted in~\cite{DBLP:conf/alenex/BergaminiMS15} show that relatively large values of $\epsilon$ (\eg $\epsilon=0.05$) lead to good ranking of nodes with high BC and for such values the number of samples is typically much smaller than $|V|$, making the memory requirements of our algorithms significantly less demanding than those of the dynamic exact algorithms ($\Omega(n^2)$) for many applications.
 \begin{theorem}
\label{thm:correctness_bc}
Algorithm~\ref{update} preserves the guarantee on the maximum absolute error, \ie naming $c'_{B}(v)$ and $\tilde{c}'_B(v)$ the new exact and approximated BC values, respectively, $\Pr(\exists v\in V\: s.t.\:|c'_{B}(v)-\tilde{c}'_B(v)|>\epsilon)<\delta$.

\end{theorem}
 \begin{theorem}
\label{thm:complexity_bc}
Let $\Delta r$ be the difference between the value of $r$ before and after the batch and let $||A^{(i)}||$ be the sum of affected nodes and their incident edges in the $i$-th SSSP. The time required for the BC update in unweighted graphs is 
$O((r+r') |\beta | + \sum_{i=1}^{r+r'} (||A^{(i)}||+d_{\max}^{(i)}) + \Delta r(|V|+|E|))$. 
In weighted graphs, it is $O((r+r') |\beta |\log |\beta | + \sum_{i=1}^{r+r'} ||A^{(i)}|| \log ||A^{(i)}|| + \Delta r(|V|\log |V|+|E|))$.
\end{theorem}
Notice that, if \vda does not increase, $\Delta r = 0$ and the complexities are the same as the only-incremental algorithms \textsf{IA} and \textsf{IAW} we proposed in~\cite{DBLP:conf/alenex/BergaminiMS15}. Also, notice that in the worst case the complexity can be as bad as recomputing from scratch. However, no dynamic SSSP (and so probably also no BC approximation) algorithm exists that is faster than recomputation.

\section{Experiments}
\label{sec:experimental}
\paragraph{Implementation and settings.} We implement our two dynamic approaches \textsf{DA} and \textsf{DAW} in C++, building on the open-source \textit{NetworKit} framework~\cite{DBLP:journals/corr/StaudtSM14}, which also contains
the static approximation \textsf{RK}.
In all experiments we fix $\delta$ to 0.1 and $\epsilon$ to 0.05, as a good tradeoff between running time and accuracy~\cite{DBLP:conf/alenex/BergaminiMS15}. 
This means that, with a probability of at least $90\%$, the computed BC values deviate at most
$0.05$ from the exact ones. In our previous experimental study~\cite{DBLP:conf/alenex/BergaminiMS15}, we showed that for such values of $\epsilon$ and $\delta$, the ranking error (how much the ranking computed by the approximation algorithm differs from the rank of the exact algorithm) is low for nodes with high betweenness. Since our algorithms simply update the approximation of \rk, our accuracy in terms or ranking error does not differ from that of \rk (see~\cite{DBLP:conf/alenex/BergaminiMS15} for details). Also, our experiments in~\cite{DBLP:conf/alenex/BergaminiMS15} have shown that dynamic exact algorithms are not scalable, because of both time and memory requirements, therefore we do not include them in our tests.
The machine used has 2 x 8 Intel(R) Xeon(R) E5-2680 cores at 2.7 GHz, of which we use only one core, and 256 GB RAM.
  \vspace{-2ex}
\begin{table*}[t]
\begin{center}
\begin{scriptsize}
  \begin{tabular}{ | l | l | r | r | r |}
    \hline
    Graph 						& Type 				& Nodes 			& Edges 			&  Type\\ \hline
   
    \texttt{repliesDigg}				& communication	& 30,398			& 85,155 			& Weighted	\\
    \texttt{emailSlashdot}			& communication	& 51,083 			& 116,573			& Weighted	\\ 
    \texttt{emailLinux}				& communication	& 63,399 			& 159,996			& Weighted	\\
    \texttt{facebookPosts}		& communication	& 46,952			& 183,412 		& Weighted	\\
    \texttt{emailEnron}				& communication	& 87,273 			& 297,456			& Weighted	\\
    \texttt{facebookFriends}			& friendship		& 63,731			& 817,035 		& Unweighted	\\
    \texttt{arXivCitations} 			& coauthorship 		& 28,093			& 3,148,447		& Unweighted	\\
    \texttt{englishWikipedia}		& hyperlink		& 1,870,709		& 36,532,531 		& Unweighted	\\

    \hline
  \end{tabular}
  \end{scriptsize}
\end{center}
  \caption{Overview of real dynamic graphs used in the experiments.}
  \label{table:graphs}
   \vspace{-4ex}
\end{table*}
\paragraph{Data sets and experiments.} We concentrate on two types of graphs: synthetic and real-world graphs with real edge dynamics. The real-world networks are taken from The Koblenz Network Collection (KONECT)~\cite{DBLP:conf/www/Kunegis13} and are summarized in Table~\ref{table:graphs}. All the edges of the KONECT graphs are characterized by a time of arrival. In case of multiple edges between two nodes, we extract two versions of the graph: one unweighted, where we ignore additional edges, and one weighted, where we replace the set  $E_{st}$ of edges between two nodes with an edge of weight $1/|E_{st}|$. 
In our experiments, we let the batch size vary from 1 to 1024 and for each batch size, we average the running times over 10 runs.
Since the networks do not include edge deletions, we implement additional simulated dynamics. In particular, we consider the following experiments. (i) \textit{Real dynamics.} We remove the $x$ edges with the highest timestamp from the network and we insert them back in batches, in the order of timestamps. (ii) \textit{Random insertions and deletions.} We remove $x$ edges from the graph, chosen uniformly at random. To create batches of both edge insertions and deletions, we add back the deleted edges with probability $1/2$ and delete other random edges with probability $1/2$. (iii) \textit{Random weight changes.} In weighted networks, we choose $x$ edges uniformly at random and we multiply their weight by a random 
value in the interval $(0,2)$.

For synthetic graphs we use a generator based on a unit-disk graph model in hyperbolic geometry~\cite{DBLP:journals/corr/LoozSMP15}, where edge insertions and deletions are obtained by moving the nodes in the hyperbolic plane. The networks produced by the model were shown to have many properties of real complex networks, like small diameter and power-law degree distribution (see~\cite{DBLP:journals/corr/LoozSMP15} and the references therein). We generate seven networks, with $|E|$ ranging from about $2\cdot 10^4$ to about $2 \cdot 10^7$ and $|V|$ approximately equal to $|E|/10$.
\paragraph{Speedups.}
\begin{figure}[th]
\begin{center}
\includegraphics[width = 0.8\textwidth]{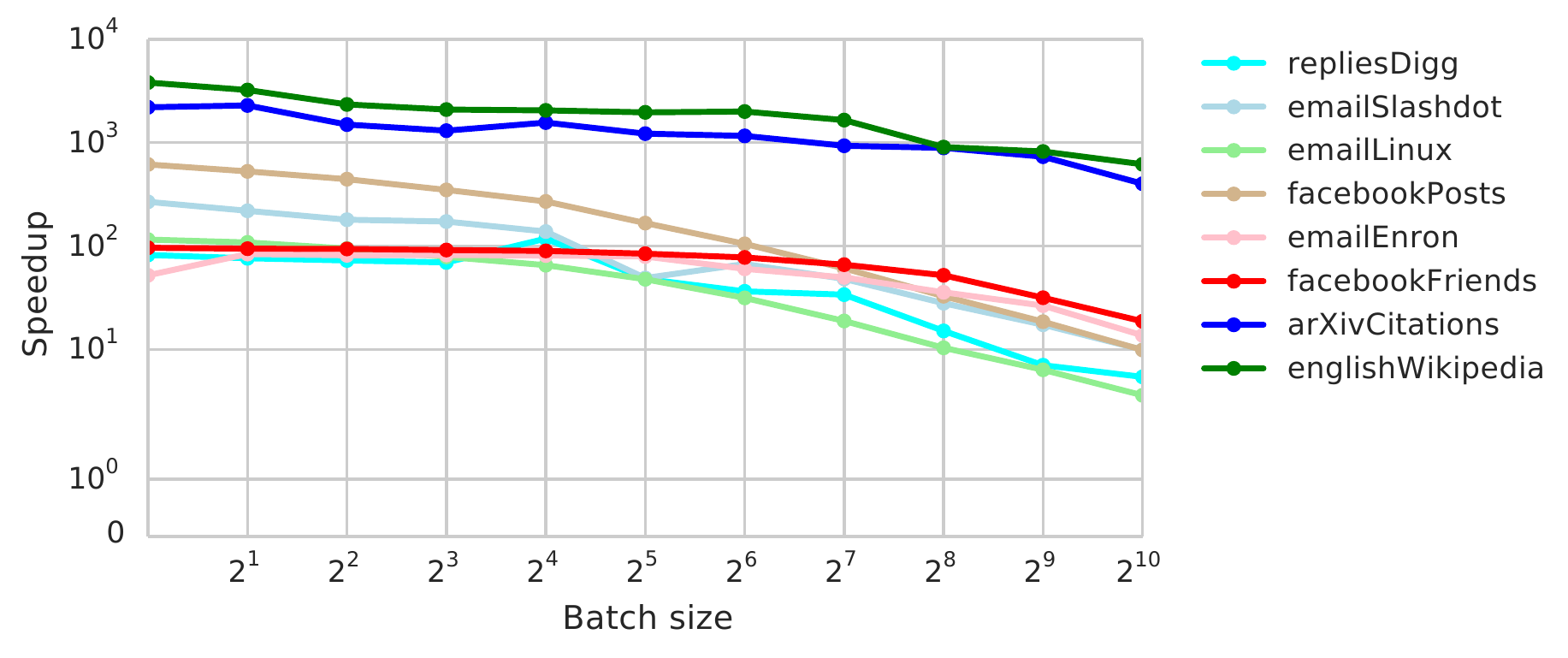}
\caption{Speedups of \da on \rk in real unweighted networks using real dynamics.}
\label{fig:unweighted}
 \vspace{-4ex}
\end{center}
\end{figure} 
 \begin{table*}
\begin{center}
\begin{scriptsize}
  \begin{tabular}{  l | r | r | r | r | r | r | r | r |}
\cline{2-9}  
 & \multicolumn{4}{  c |  }{Real}& \multicolumn{4}{  c |  }{Random} \\ \cline{2-9}
 &\multicolumn{2}{  c |  }{Time [s]}& \multicolumn{2}{  c |  }{Speedups} &\multicolumn{2}{  c |  }{Time [s]}& \multicolumn{2}{  c |  }{Speedups} \\ \cline{1-9}
\multicolumn{1}{| l|}{Graph} & $|\beta|=1$ & $|\beta|=1024$ & $|\beta|=1$ & $|\beta|=1024$  & $|\beta|=1$ & $|\beta|=1024$ & $|\beta|=1$ & $|\beta|=1024$ \\\cline{1-9}
\multicolumn{1}{| l|}{\texttt{repliesDigg}} 		& 0.078	& 1.028	& 76.11	& 5.42	& 0.008	& 0.832	& 94.00	& 4.76	\\ \cline{1-9}		
\multicolumn{1}{| l|}{\texttt{emailSlashdot}} 	& 0.043	& 1.055	& 219.02	& 9.91	& 0.038	& 1.151	& 263.89	& 28.81	\\ \cline{1-9}		
\multicolumn{1}{| l|}{\texttt{emailLinux}} 		& 0.049	& 1.412	& 108.28	& 3.59	& 0.051	& 2.144	& 72.73	& 1.33	\\ \cline{1-9}				
\multicolumn{1}{| l|}{\texttt{facebookPosts}} 	& 0.023	& 1.416	& 527.04	& 9.86	& 0.015	& 1.520	& 745.86	& 8.21	\\ \cline{1-9}	
\multicolumn{1}{| l|}{\texttt{emailEnron}} 		& 0.368	& 1.279	& 83.59	& 13.66	& 0.203	& 1.640	& 99.45	& 9.39	\\ \cline{1-9}				
\multicolumn{1}{| l|}{\texttt{facebookFriends}} 	& 0.447	& 1.946	& 94.23	& 18.70	& 0.448	& 2.184	& 95.91	& 18.24	\\ \cline{1-9}			
\multicolumn{1}{| l|}{\texttt{arXivCitations}} 	& 0.038	& 0.186	& 2287.84	& 400.45	& 0.025	& 1.520	& 2188.70	& 28.81	\\ \cline{1-9}			
\multicolumn{1}{| l|}{\texttt{englishWikipedia}} 	& 1.078	& 6.735	& 3226.11	& 617.47	& 0.877	& 5.937	& 2833.57	& 703.18	\\ \cline{1-9}		
  \end{tabular}
  \end{scriptsize}
\end{center}
  \caption{Times and speedups of \da on \rk in unweighted real graphs under real dynamics and random updates, for batch sizes of 1 and 1024.}
  \label{table:speedups}
   \vspace{-6ex}
\end{table*} 

\begin{figure}[htb]
\begin{center}
\includegraphics[width = 0.8\textwidth]{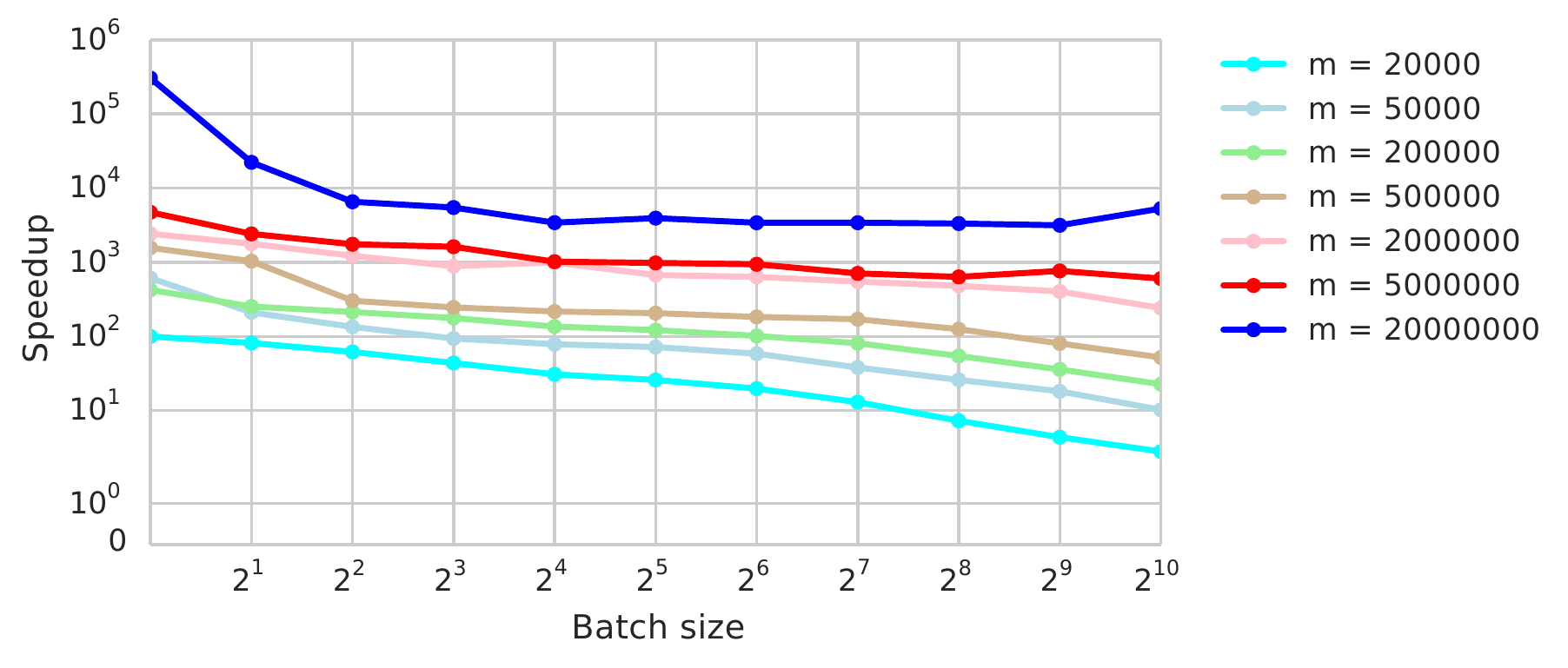}
\caption{Speedups of \da on \rk in hyperbolic unit-disk graphs.}
\label{fig:hyp}
\end{center}
\end{figure}
Figure~\ref{fig:unweighted} reports the speedups of \da on \rk in real graphs using real dynamics. Although some fluctuations can be noticed, the speedups tend to decrease as the batch size increases. We can attribute fluctuations to two main factors: First, different batches can affect areas of $G$ of varying sizes, influencing also the time required to update the SSSPs. Second, changes in the \vd approximation can require to sample new paths and therefore increase the running time of \da (and \daw). Nevertheless, \da is significantly faster than recomputation on all networks and for every tested batch size. Analogous results are reported in Figure~\ref{fig:random_unweighted} of the Appendix for random dynamics. Table~\ref{table:speedups} summarizes the running times of \da and its speedups on \rk with batches of size 1 and 1024 in unweighted graphs, under both real and random dynamics. Even on the larger graphs (\texttt{arXivCitations} and \texttt{englishWikipedia}) and on large batches, \da requires at most a few seconds to recompute the BC scores, whereas \rk requires about one hour for \texttt{englishWikipedia}. 
The results on weighted graphs are shown in Table~\ref{table:weighted} in Section~\ref{sec:plots} in the Appendix. In both real dynamics and random updates, the speedups vary between $\approx 50$ and $\approx 6 \cdot 10^3$ for single-edge updates and between $\approx 5$ and $\approx 75$ for batches of size 1024.
On hyperbolic graphs (Figure~\ref{fig:hyp}), the speedups of \da on \rk increase with the size of the graph. Table~\ref{table:hyperbolic} in the Appendix contains the running times and speedups on batches of 1 and 1024 edges. The speedups vary between $\approx 100$ and $\approx 3 \cdot 10^5$ for single-edge updates and between $\approx 3$ and $\approx 5 \cdot 10^3$ for batches of 1024 edges. 
The results show that \da and \daw are faster than recomputation with \rk in all the tested instances, even when large batches of 1024 edges are applied to the graph. With small batches, the algorithms are always orders of magnitude faster than \rk, often with running times of fraction of seconds or seconds compared to minutes or hours. Such high speedups are made possible by the efficient update of the sampled shortest paths, which limit the recomputation to the nodes that are actually affected by the batch. Also, processing the edges in batches, we avoid to update multiple times nodes that are affected by several edges of the batch.

%
%
\section{Conclusions}
%
Betweenness is a widely used centrality measure, yet expensive if computed exactly.
In this paper we have presented the first fully-dynamic algorithms for betweenness approximation  
(for weighted and for unweighted undirected graphs). 
The consideration of edge deletions and disconnected graphs is made possible by the efficient solution
of several algorithmic subproblems (some of which may be of independent interest).
Now BC can be approximated with an error guarantee for a much wider 
set of dynamic real graphs compared to previous work. 
%

Our experiments show significant speedups over the static algorithm \rk. In this context it is interesting
to remark that dynamic algorithms require to store additional memory and that this can be a limit to the size 
of the graphs they can be applied to. 
By not storing the predecessors in the shortest paths, we reduce the memory requirement from $\Theta(|E|)$ per 
sampled path to $\Theta(|V|)$ -- and are still often more than 100 times faster than \rk despite rebuilding the paths.

Future work may include the transfer of our concepts to approximating other centrality measures in a fully-dynamic
manner, \eg closeness, and the extension to directed graphs, for which a good \vd approximation is the only obstacle.
Moreover, making the betweenness code run in parallel will further accelerate the
computations in practice.
Our implementation will be made available as part of a future release of the network analysis tool suite 
\href{http://networkit.iti.kit.edu}{\textit{NetworKit}}~\cite{DBLP:journals/corr/StaudtSM14}.

\bigskip

\renewcommand{\baselinestretch}{0.25}
\begin{scriptsize}
\textbf{Acknowledgements.}
This work is partially supported by DFG grant FINCA (ME-3619/3-1) within the SPP 1736 \emph{Algorithms for Big Data}.
We thank Moritz von Looz for providing the synthetic dynamic networks and the numerous contributors to 
the \textit{NetworKit} project. We also thank Matteo Riondato (Brown University) and anonymous reviewers for their constructive comments.
\end{scriptsize}
\renewcommand{\baselinestretch}{1.0}

\bibliographystyle{abbrv}
\bibliography{references}

\pagebreak
\appendix

\section{Description of the fully-dynamic algorithms}
\label{sec:pseudocodes}

\renewcommand{\floatpagefraction}{0.99}
\subsection{Dynamic \vd approximation}
Algorithm~\ref{algo:vd_approx} describes the initialization. Initially, we put all the nodes in a queue and compute an SSSP from the nodes we extract. During the SSSP search, we mark as visited all the nodes we scan. When extracting the nodes, we skip those that have already been visited: this avoids us to compute multiple approximations for the same component. 
In the update (Algorithm~\ref{algo:dyn_vd_approx}), we recompute the SSSPs and the \vd approximations with \upvd (or \upvdu). Since components might split, we might need to add \vd approximations for some new subcomponents, in addition to recompute the old ones. Also, if components merge, we can discard the superfluous approximations. To do this, we keep track, for each node, of the number of times it has been visited. Let $vis(v)$ denote this number for node $v$. Before the update, all the nodes are visited exactly once. While updating an SSSP from $s_i$, we increase (decrease) by one $vis(v)$ of the nodes $v$ that become reachable (unreachable) from $s_i$. This way we can skip the update of the SSSPs from nodes that have already been visited. After the update, for all nodes $v$ that have become unvisited ($vis(v)=0$), we compute a new \vd approximation from scratch.

\begin{algorithm}[h]
 \begin{footnotesize}
\LinesNumbered
\SetKwFunction{initDynamicSSSP}{initApprVD}
$U \leftarrow []$\;
\ForEach{node $v \in V$}
{
		$vis(v)\leftarrow 0$; insert $v$ into $U$\; \label{initvd:queue}
}
$i \leftarrow 1$\;
\While{$U\neq \emptyset$}
{\label{initvd:newsamples1}
	extract $s$ from $U$\;
	\If{$vis(s) = 0$}
	{
		 $s_i \leftarrow s$\;
		 \tcp{\initDynamicSSSP adds 1 to $vis(v)$ of the nodes it visits}
		$\tilde{\textsf{VD}}_i$ $\leftarrow$ \initDynamicSSSP{$G,s_i$}\;
		$i \leftarrow i+1$\;
	}
}\label{initvd:newsamples2}
$n_C \leftarrow i-1$\;
\vda $\leftarrow \max_{i=1,...,n_C} \tilde{\mathit{VD}}_i$\;
\Return{\vda}
\end{footnotesize}
\caption{Dynamic \vd approximation (initialization)}
\label{algo:vd_approx}
\end{algorithm}

\begin{algorithm}[h]
 \begin{footnotesize}
\LinesNumbered
\SetKwFunction{dynamicSSSP}{updateApprVD}
\SetKwFunction{initDynamicSSSP}{initApprVD}
$U \leftarrow []$\; 
\ForEach{$s_i$}
{\label{vdup:loop1}
	\If{vis($s_i )> 1$}
	{	
		remove $s_i$ and $\tilde{\mathit{VD}}_i$; decrease $n_C$\; \label{vdup:skip}
	}
	\Else{
		 \tcp{\dynamicSSSP updates $vis$, inserts all $v$ for which $vis(v)=0$ into $U$ and computes a \vd approximation $\tilde{\mathit{VD}}_i$}
		 $\tilde{\textsf{VD}}_i$ $\leftarrow$ \dynamicSSSP{$G,s_i$} \;
	}
}\label{vdup:loop2}
$i \leftarrow n_C$ \;
\While{$U\neq \emptyset$}
{\label{vdup:ext1}
	extract $s'$ from $U$\; \label{vdup:extract}
	\If{$vis(s') = 0$}
	{
		 $s'_i \leftarrow s'$\;
		$\tilde{\textsf{VD}}_i$ $\leftarrow$ \initDynamicSSSP{$G,s'_i$}\;
		$i \leftarrow i+1$; $n_C \leftarrow n_C+1$\;
	}
}\label{vdup:ext2}
reset $vis(v)$ to 1 for nodes $v$ such that $vis(v)>1$\;
\vda $\leftarrow \max_{i=1,...,n_C} \tilde{\mathit{VD}}_i$\;
\Return{\vda}
\end{footnotesize}
\caption{Dynamic \vd approximation (\textsf{updateApprVD})}
\label{algo:dyn_vd_approx}
\end{algorithm}

\subsection{Dynamic SSSP update for weighted graphs}
Algorithm~\ref{algo:weighted} describes the SSSP update for weighted graphs. The pseudocode updates both the \vd approximation for the connected component of $s$ and the number of shortest paths from $s$, so it basically includes both \sssp and \upvd.
Initially, we scan the edges $e =\{u,v\}$ in $\beta$ and, for each $e$, we insert the endpoint with greater distance from $s$ into $Q$ (w.l.o.g., let $v$ be such endpoint). The priority $p(v)$ of $v$ represents the \textit{candidate} new distance of $v$. This is the minimum between the $d(v)$ and $d(u)$ plus the weight of the edge $\{u, v\}$. Notice that we use the expression "insert $v$ into $Q$" for simplicity, but this can also mean update $p(v)$ if $v$ is already in $Q$ and the new priority is smaller than $p(v)$. 
When we extract a node $w$ from $Q$, we have two possibilities: (i) there is a path of length $p(w)$ and $p(w)$ is actually the new distance or (ii) there is no path of length $p(w)$ and the new distance is greater than $p(w)$. In the first case (Lines~\ref{update_weighted:part1-1}~-~\ref{update_weighted:part1-2}), we set $d(w)$ to $p(w)$ and insert the neighbors $z$ of $w$ such that $d(z)>d(w)+\omega(\{w,z\})$ into $Q$ (to check if new shorter paths to $z$ that go through $w$ exist). In the second case (Lines~\ref{update_weighted:part2-1}~-~\ref{update_weighted:part2-2}), we assume there is no shortest path between $s$ and $w$ anymore, setting $d(w)$ to $\infty$. We compute $p(w)$ as $\min_{\{v,w\}\in E} d(v) + \omega{(v,w)}$ (the new candidate distance for $w$) and insert $w$ into $Q$. Also its neighbors could have lost one (or all of) their old shortest paths, so we insert them into $Q$ as well. The update of $\underline{\omega}$ can be done while scanning the batch and of $d'$ and $d''$ when we update $d(w)$. When updating $d(w)$, we also increase $vis(w)$ in case the old $d(w)$ was equal to $\infty$ (\ie w has become reachable) and we decrease $vis(w)$ when we set $d(w)$ to $\infty$ (\ie $w$ has become unreachable). We update the number of shortest paths after updating $d(w)$, as the sum of the shortest paths of the predecessors of $w$ (Lines~\ref{update_weighted:sp1}~-~\ref{update_weighted:sp2}).
\begin{algorithm}[H]
\begin{footnotesize}
\LinesNumbered
\SetKwFunction{insertOrDecreaseKey}{insertOrDecreaseKey}
\SetKwFunction{extractMin}{extractMin}
$Q\leftarrow$ empty priority queue\;
\ForEach{$e=\{u,v\} \in \beta, d(u)<d(v)$}
{\label{update_weighted:init1}
	$Q \leftarrow$ \insertOrDecreaseKey{$v, p(v)=\min\{d(u)+\omega(\{u,v\}),d(w)\}$}\;
}\label{update_weighted:init2}
$\underline{\omega} \leftarrow \min\{\underline{\omega},\  \omega(e): e \in \beta\}$\;
\While{there are nodes in $Q$}
{
		$\{w, p(w)\} \leftarrow$ \extractMin{$Q$}\;
		$con(w) \leftarrow \min_{z : (z,w)\in E} d(z) + \omega(z,w)$\;
		\If{$con(w) = p(w)$}
		{\label{update_weighted:part1-1}
			update $d'$ and $d''$\;
			\If{$d(w) = \infty$} 
			{
				$vis(w) \leftarrow vis(w)+1$\; \label{update_weighted:visinc}
			}
			$d(w) \leftarrow p(w)$; $\sigma(w)\leftarrow 0$\;
			\ForEach{incident edge $(z,w)$}
			{
				\If{$d(w)=d(z)+\omega(z,w)$}
				{\label{update_weighted:sp1}
					$\sigma(w)\leftarrow \sigma(w)+\sigma(z)$\;
				}\label{update_weighted:sp2}
				\If{$d(z)\geq d(w)+\omega(z,w)$}
				{\label{update_w:neigh1}
					$Q \leftarrow $ \insertOrDecreaseKey{$z, p(z)=d(w)+\omega(z,w)$}\;
				}\label{update_w:neigh2}
			}
		}\label{update_weighted:part1-2}
		\Else
		{\label{update_weighted:part2-1}
			\If{$d(w) \neq \infty$}
			{
				$vis(w) \leftarrow vis(w)-1$\; \label{update_weighted:visdec}
				\If{vis(w)=0}{
				insert $w$ into $U$\;
				}
				\If{$con(w) \neq \infty$}
				{
					$Q \leftarrow $\insertOrDecreaseKey{$w, p(w)=con(w)$}\;
					\ForEach{incident edge $(z,w)$}
					{
						\If{$d(z)=d(w)+\omega(w,z)$}
						{
							$Q \leftarrow $\insertOrDecreaseKey{$z,p(z)=d(w)+\omega(z,w)$}\;
						}
					}
					$d(w)\leftarrow \infty$\;
				}
			}
		}\label{update_weighted:part2-2}
}
\end{footnotesize} 
\caption{SSSP update for weighted graphs (\textsf{updateSSSP-W})}
\label{algo:weighted}
\end{algorithm}
 \vspace{4ex}
\begin{algorithm}[H]
\begin{footnotesize}
\LinesNumbered
\textbf{Assumption:} $color(w)=white\quad \forall w \in V$\;
$Q[]\leftarrow$ array of empty queues\;
\ForEach{$e=\{u,v\} \in \beta, d(u)<d(v)$}
{\label{update:batch1}
	$k \leftarrow d(v)+1$; enqueue $v \rightarrow Q[k]$\;
}\label{update:batch2}
$k \leftarrow 1$\; \label{update:queues1}
\While{there are nodes in $Q[j],  j\geq k$}
{\label{lst2:line:start2}
	\While{$Q[k]\neq \emptyset$}
	{
		dequeue $w\leftarrow Q[k]$\;
		\textbf{if} $color(w)=black$ \textbf{then}			continue\; \label{update:black2}
	
		$con(w) \leftarrow \min_{z : (z,w)\in E} d(z) + 1$\;
		\If{$con(w) = k$}
		{\label{update:if}
			update $d'$ and $d''$\;
			\textbf{if} $d(w) = \infty$  \textbf{then}   $vis(w) \leftarrow vis(w)+1$\; \label{update:vis}
			$d(w) \leftarrow k$; $\sigma(w)\leftarrow 0$; $color(w) \leftarrow black$\; \label{update:black1}
			\ForEach{incident edge $(z,w)$}
			{
				\If{$d(w)=d(z)+1$}
				{\label{update:sp}
					$\sigma(w)\leftarrow \sigma(w)+\sigma(z)$\;
				}\label{update:sp2}
				\If{$d(z)>k$}
				{\label{update:neigh1}
					enqueue $z \rightarrow Q[k+1]$\;
				}\label{update:neigh2}
			}
		}\label{update:if2}
		\Else
		{\label{update:else}
			\If{$d(w) \neq \infty$}
			{
				$d(w)\leftarrow \infty$\;
				$vis(w) \leftarrow vis(w)-1$\; \label{update:vis3}
				\If{vis(w)=0}{
				insert $w$ into $U$\;
				}
				\If{$con(w) \neq \infty$}
				{
					enqueue $w \rightarrow Q[con(w)]$\;
					\ForEach{incident edge $(z,w)$}
					{
						\If{$d(z)>k$}
						{
							enqueue $z \rightarrow Q[k+1]$\;
						}
					}
				}
			}
		}\label{update:else2}
	}
	$k\leftarrow k+1$\;
}\label{update:queues2}
Set to white all the nodes that have been in $Q$\;
\end{footnotesize}
\caption{SSSP update for unweighted graphs (\textsf{updateSSSP-U})}
\label{algo:unweighted}
\end{algorithm}
\subsection{Dynamic SSSP update for unweighted graphs}
Algorithm~\ref{algo:unweighted} shows the pseudocode. As in Algorithm~\ref{algo:weighted}, we
first scan the batch (Lines~\ref{update:batch1} -~\ref{update:batch2}) and insert the nodes
in the queues. Then (Lines~\ref{update:queues1} -~\ref{update:queues2}), we scan the queues in order of increasing distance from $s$, in a fashion similar to that of a priority queue. 
In order not to insert a node in the queues multiple times, we use colors: Initially we set all the nodes to white and then we set a node $w$ to black only when we find the final distance of $w$ (\ie when we set $d(w)$ to $k$) (Line~\ref{update:black1}). Black nodes extracted from a queue are then skipped (Line~\ref{update:black2}). At the end we reset all nodes to white.

\subsection{Fully-dynamic BC approximation} 
\label{sec:fully_dyn_bc}
Similarly to \textsf{IA} and \textsf{IAW}, we replace the $r$ sampled paths between vertex pairs $(s,t)$ with new shortest paths between the same vertex pairs. However, here we also check whether $\tilde{\mathit{VD}}$ (and consequently the number $r$ of samples) has increased after the batch of edge updates. If so, we sample additional paths (computing new SSSPs from scratch) according to the new value of $r$. Instead of updating $\tilde{\mathit{VD}}$ and then the paths in two successive steps, we use the SSSPs from the $r$ source nodes $s$ to compute and update also $\tilde{\mathit{VD}}$, computing new SSSPs only for the components that are not covered by any of the source nodes. In the initialization (Algorithm~\ref{algo:init}), we first compute the $r$ SSSP, like in \rk (Lines~\ref{init:sampling1}~-~\ref{init:sampling2}). However, we also check which nodes have been visited, as in Algorithm~\ref{algo:vd_approx}. While we compute the $r$ SSSPs, in addition to the distances and number of shortest paths, we also compute a \vd approximation for each of the $r$ source nodes and increase $vis(v)$ of all the nodes we visit during the sources with \textsf{initSSSPVD} (Line~\ref{init:dynsssp}). Since it is possible that the $r$ shortest paths do not cover all the components of $G$, we compute an additional VD approximation for nodes in the unvisited components, like in Algorithm~\ref{algo:vd_approx} (Lines~\ref{init:newsamples1}~-~\ref{init:newsamples2}). Basically we can divide the SSSPs into two sets: the set $R$ of SSSPs used to compute the $r$ shortest paths and the set $R'$ of SSSPs used for a \vd approximation in the components that were not scanned by the initial $R$ SSSPs. We call $r'$ the number of the SSSPs in $R'$. The BC update after a batch is described in Algorithm~\ref{update}. First (Lines~\ref{up:update1} -~\ref{up:update2}), we recompute the shortest paths like in our incremental algorithms \textsf{IA} and \textsf{IAW}~\cite{DBLP:conf/alenex/BergaminiMS15}: we update the SSSPs from each source node $s$ in $R$ and we replace the old shortest path with a new one (subtracting $1/r$ to the nodes in the old shortest path and adding $1/r$ to those in the new shortest path). Notice that here we do not store the predecessors so we need to recompute them (Lines~\ref{up:pred1} and~\ref{up:pred2}).
Instead of using an incremental SSSP algorithm like in \textsf{IA}-\textsf{IAW}, here we use the fully-dynamic \textsf{updateSSSPVD} that updates also the \vd approximation and updates and keeps track of the nodes that become unvisited.
Then (Lines~\ref{up:ext1} -~\ref{up:ext2}), we add a new SSSP to $R'$ for each component that has become unvisited (by both $R$ and $R'$). After this, we have at least a \vd approximation for each component of $G$. We take the maximum over all these approximations and recompute the number of samples $r$ (Lines~\ref{up:recompute1} -~\ref{up:recompute2}). If $r$ has increased, we need to sample new paths and therefore new SSSPs to add to $R$. Finally, we normalize the BC scores, \ie we multiply them by the old value of $r$ divided by the new value of $r$ (Line~\ref{up:norm}).

\begin{algorithm}[h]
 \begin{footnotesize}
\LinesNumbered
\SetKwData{B}{$\tilde{c}_B$}\SetKwData{VD}{VD}
\SetKwFunction{getVertexDiameter}{getApproxVertexDiameter}
\SetKwFunction{sampleUniformNodePair}{sampleUniformNodePair}
\SetKwFunction{computeExtendedSSSP}{computeExtendedSSSP}
\SetKwFunction{initDynamicSSSP}{initApprVD}
\SetKwFunction{initSSSPVD}{initSSSPVD}
\ForEach{node $v \in V$}
{
	\B$(v)\leftarrow 0$; $vis(v) \leftarrow 0$\; \label{init:init}
}
\vda$\leftarrow$\getVertexDiameter{$G$}\; \label{init:sampling1}
$r \leftarrow (c/\epsilon^2) (\lfloor \log_2(\tilde{\mathit{VD}}-2)\rfloor +\ln(1/\delta))$\;

\For{$i \leftarrow 1$ \KwTo $r$}{
	\label{sampling1}
	$(s_i,t_i)\leftarrow$ \sampleUniformNodePair{$V$}\;
	$\tilde{\mathit{VD}}_i \leftarrow$ \initSSSPVD{$G,s_i,$}\; \label{init:dynsssp}
	$v \leftarrow t_i$\;\label{paths1}
	$p_{(i)}\leftarrow$ empty list\;
	$P_{s_i}(v) \leftarrow\{ z : \{z,v\}\in E \cap d_{s_i}(v) = d_{s_i}(z)+\omega(\{z,v\}) \}$\;
	\While{$P_{s_i}(v) \neq \{s_i\}$}
	{
		\mbox{sample $z \in P_{s_i}(v)$ with probability $\sigma_{s_i}(z)/\sigma_{s_i}(v)$}\;
		$\B(z)\leftarrow \B(z)+1/r$\;
		add $z\rightarrow p_{(i)}$;
		$v\leftarrow z$\;
		$P_{s_i}(v) \leftarrow\{ z : \{z,v\}\in E \cap d_{s_i}(v) = d_{s_i}(z)+\omega(\{z,v\}) \}$\;
	}\label{paths2}
}\label{sampling2} \label{init:sampling2}
$U \leftarrow V$\;
$i \leftarrow r+1$\;
\While{$U\neq \emptyset$}
{\label{init:newsamples1}
	extract $s'$ from $U$\;
	\If{$vis(s') = 0$}
	{
		 $s'_i \leftarrow s'$\;
		$\tilde{\mathit{VD}}_i \leftarrow$ \initDynamicSSSP{$G,s'_i$}\;
		$i \leftarrow i+1$\;
	}
}\label{init:newsamples2}
$r' \leftarrow r-i-1$\;
\Return{$\{(v,\B(v)):\: v\in V\}$}
\end{footnotesize}
\caption{BC initialization}
\label{algo:init}
\end{algorithm}

\begin{algorithm}[H]
\begin{footnotesize}
\LinesNumbered
\SetKwData{B}{$\tilde{c}_B$}
\SetKwFunction{dynamicSSSP}{updateApprVD}
\SetKwFunction{dynamicSSSPVD}{updateSSSPVD}
\SetKwFunction{initDynamicSSSP}{initApprVD}
$U \leftarrow []$\;
\For{$i \leftarrow 1$ \KwTo $r$}
{\label{up:update1}
	$d^{old}_i\leftarrow d_{s_i}(t_i)$\;
	$\sigma^{old}_i\leftarrow \sigma_{s_i}(t_i)$\;
	 \tcp{\dynamicSSSPVD updates $vis$, inserts all $v :\: vis(v)=0$ into $U$ and updates the \vd approximation}
	$\tilde{\mathit{VD}}_i \leftarrow$ \dynamicSSSPVD{$G, s_i, \beta$}\; \label{up:dynsssp}
	\tcp { we replace the shortest path between $s_i$ and $t_i$}
			\ForEach{$w \in p_{(i)}$}
			{
				\B($w$) $\leftarrow \B(w)-1/r$\; \label{up:sub}
			}
			$v \leftarrow t_i$\;
			$p_{(i)}\leftarrow$ empty list\;
			$P_{s_i}(v) \leftarrow\{ z : \{z,v\}\in E \cap d_{s_i}(v) = d_{s_i}(z)+\omega(\{z,v\}) \}$\;\label{up:pred1}
			\While{$P_{s_i}(v) \neq \{s_i\}$}
			{\label{up:new1}
				\mbox{sample $z \in P_{s_i}(v)$ with probability $=\sigma_{s_i}(z)/\sigma_{s_i}(v)$}\;
				$\B(z)\leftarrow \B(z)+1/r$\;
				add $z$ to $p_{(i)}$\;
				$v\leftarrow z$\;
				$P_{s_i}(v) \leftarrow\{ z : \{z,v\}\in E \cap d_{s_i}(v) = d_{s_i}(z)+\omega(\{z,v\}) \}$\;\label{up:pred2}
			}\label{up:new2}
}\label{up:update2}
\For{$i \leftarrow r+1$ \KwTo $r+r'$}
{\label{up:update3}
		$\tilde{\mathit{VD}}_i \leftarrow$ \dynamicSSSP{$G, s_i, \beta$}\;
}\label{up:update4}
$i \leftarrow r+r'+1$\;\label{ext1}
\While{$U\neq \emptyset$}
{\label{up:ext1}
	extract $s'$ from $U$\;
	\If{$vis(s') = 0$}
	{
		 $s'_i \leftarrow s'$\;
		$\tilde{\mathit{VD}}_i \leftarrow$ \initDynamicSSSP{$G, s'_i$}\;
		$i \leftarrow i+1$; $r' \leftarrow r'+1$\;
	}
}\label{up:ext2}
\tcp{compute the maximum over all the ${\mathit{VD}}_i$ computed by \dynamicSSSP}
\vda $\leftarrow \max_{i=1,...,r+r'} \tilde{\mathit{VD}}_i$\; \label{up:recompute1}
$r_{\text{new}} \leftarrow (c/\epsilon^2) (\lfloor \log_2(\tilde{\mathit{VD}}-2)\rfloor +\ln(1/\delta))$\; \label{up:recompute2}
\If{$r_{\text{new}}>r$}
{
	sample new paths\; \label{up:new_paths}
	\ForEach{$v \in V$}
	{
		$\B(v) \leftarrow \B(v) \cdot r/r_{\text{new}}$ \label{up:norm}
	}
$r \leftarrow r_{\text{new}}$\;
}\label{ext4}
\Return{$\{(v,\B(v)):\: v\in V\}$}
\end{footnotesize}
\caption{Dynamic update of BC approximation (\textsf{DA})}
\label{update}
\end{algorithm}

\section{Omitted proofs}
\label{sec:proofs}
\subsection{Proof of Proposition~\ref{lem:vd2}}
\label{sub:proof_vd2}
\begin{proof}
To prove the first inequality, we can notice that $d^T(x,y) \geq d(x,y)$ for all $x,y \in V$, since all the edges of $T$ are contained in those of $G$. Also, since every edge has weight at least $\underline{\omega}$, $d(x,y) \geq (|p_{xy}|-1)\cdot \underline{\omega}$. Therefore, $d^T(x,y) \geq (|p_{xy}|-1)\cdot \underline{\omega}$, which can be rewritten as $|p_{xy}| \leq 1+ \frac{d^T(x,y)}{\underline{\omega}}$, for all $x,y \in V$. Thus, $\mathit{VD} = \max_{x,y} |p_{xy}| \leq 1+(\max_{x,y} d^T(x,y))/ \underline{\omega} \leq 1 + \frac{d^T(s,u)+d^T(s,v)}{\underline{\omega}} = 1 + \frac{d(s,u)+d(s,v)}{\underline{\omega}}$, where the last expression equals \vda by definition.

To prove the second inequality, we first notice that $d(s,u) \leq (|p_{su}|-1)\cdot \overline{\omega}$,
 and analogously $d(s,v) \leq (|p_{sv}|-1)\cdot \overline{\omega}$. Consequently, $\tilde{\mathit{VD}} \leq 1+(|p_{su}| + |p_{sv}|-2)\cdot \frac{\overline{\omega}}{\underline{\omega}} < 2 \cdot |p_{su}|\cdot \frac{\overline{\omega}}{\underline{\omega}}$, supposing that $|p_{su}| \geq |p_{sv}|$ without loss of generality. By definition of \vd, $|p_{su}| \leq \mathit{VD}$. Therefore, $\tilde{\mathit{VD}} < 2 \cdot \mathit{VD} \cdot \frac{\overline{\omega}}{\underline{\omega}}$.
 \qed
\end{proof}

\subsection{Proof of Lemma~\ref{thm:complexity}}
\label{sub:proof_complexity}
\begin{proof}
In the initial scan of the batch (Lines~\ref{update_weighted:init1}-\ref{update_weighted:init2}), we scan the nodes of the batch and insert the affected nodes into $Q$ (or update their value). This requires at most one heap operation (insert or decrease-key) for each element of $\beta$, therefore $O(|\beta| \log|\beta|)$ time.
When we extract a node $w$ from $Q$, we have two possibilities: (i) $con(w)=p(w)$ (Lines~\ref{update_weighted:part1-1}~-~\ref{update_weighted:part1-2}) or (ii) $con(w)>p(w)$ (Lines~\ref{update_weighted:part2-1}~-~\ref{update_weighted:part2-2}). In the first case, we scan the neighbors of $w$ and perform at most one heap operation for each of them (Lines~\ref{update_w:neigh1}~-~\ref{update_w:neigh2}). In the second case, this happens only if $d(w)\neq \infty$. Therefore, we can perform up to one heap operation per incident edge of $w$, for each extraction of $w$ in which $d(w)\neq \infty$ or $con(w)=p(w)$.
How many times can an affected node $w$ be extracted from $Q$ with $d(w) \neq \infty$ or $con(w)=p(w)$? If the first time we extract $w$, $con(w) $ is equal to $p(w)$ (case (i)), then the final value of $d(w)$ is reached and $w$ is not inserted into $Q$ anymore. If the first time we extract $w$, $con(w)$ is greater than $p(w)$ (case (ii)), $w$ can be inserted into the queue again. However, his distance is set to $\infty$ and therefore no additional operations are performed, until $d(w)$ becomes less than $\infty$. But this can happen only in case (i), after which $d(w)$ reaches its final value. To summarize, each affected node $w$ can be extracted from $Q$ with $d(w)\neq \infty$ or $con(w)=p(w)$ at most twice and, every time this happens, at most one heap operation per incident edge of $w$ is performed. The complexity is therefore $O(|\beta |\log |\beta | + ||A || \log ||A ||)$. \qed
\end{proof}

\subsection{Proof of Lemma~\ref{thm:complexity2}}
\label{sub:proof_complexity2}
\begin{proof}
The complexity of the initialization (Lines~\ref{update:batch1}~-~\ref{update:batch2}) of Algorithm~\ref{algo:unweighted}
is $O(|\beta|)$, as we have to scan the batch. In the main loop (Lines~\ref{update:queues1}~-~\ref{update:queues2}), we scan all the list
of queues, whose final size is $d_{\max}$. 
Every time we extract a node $w$ whose color is not black, we scan all the 
incident edges, therefore this operation is linear in the number of neighbors of $w$. 
If the first time we extract $w$ (say at level $k$) $con(w) $ is equal to $k$, then $w$ will be set to black and will not be scanned anymore.
If the first time we extract $w$, $con(w)$ is instead greater than $k$, $w$ will be inserted into the queue at level $con(w)$ (if $con(w)<\infty$). 
Also, other inconsistent neighbors of $w$ might insert $w$ in one of the queues. However, after the first time $w$ is extracted,
its distance is set to $\infty$, so its neighbors will not be scanned unless $con(w)=k$, in which case 
they will be scanned again, but for the last time, since $w$ will be set to black. To summarize, each affected node and its neighbors can 
be scanned at most twice. The complexity of the
algorithm is therefore $O(|\beta|+\|A\|+d_{\max})$. \qed
\end{proof}

\subsection{Correctness of Algorithm~\ref{algo:vd_approx} and Algorithm~\ref{algo:dyn_vd_approx}}
\begin{lemma}
\label{thm:vd_correctness}
At the end of Algorithm~\ref{algo:vd_approx}, $vis(v)=1,\ \forall v \in V$ and exactly one \vd approximation is computed for each connected component of $G$.
\end{lemma}
\begin{proof}
Let $v$ be any node. Then $v$ must be scanned by \emph{at least} one source node $s_i$ in the while loop (Lines~\ref{initvd:newsamples1}~-~\ref{initvd:newsamples2}): In fact, either $v$ is visited by some $s_i$ before $v$ is extracted from $U$, or $vis(v)=0$ at the moment of the extraction and $v$ becomes a source node itself. This implies that $vis(v) \geq 1, \ \forall v \in V$.
On the other hand, $vis(v)$ cannot be greater than $1$. In fact, let us assume by contradiction that $vis(v)>1$. This means that there are at least two source nodes $s_i$ and $s_j$ ($i<j$, w.l.o.g.) that are in the same connected component as $v$. Then also $s_i$ and $s_j$ are in the same connected component and $s_j$ is visited during the SSSP search from $s_i$. Then $vis(s_j)=1$ before $s_j$ is extracted from $U$ and $s_j$ cannot be a source node.
Therefore, $vis(v)$ is exactly equal to 1 for each $v \in V$, which means that exactly one \vd approximation is computed for the connected component of each $v$, \ie exactly one \vd approximation is computed for each connected component of $G$.
 \qed
\end{proof}

\begin{lemma}
\label{thm:dyn_vd_correctness}
Let $C' = \{C'_1,...,C'_{n'_c}\}$ be the set of connected components of $G$ after the update. Algorithm~\ref{algo:dyn_vd_approx} updates or computes exactly one \vd approximation for each $C'_i \in C'$.
\end{lemma}
\begin{proof}
Let $C = \{C_1,...,C_{n_c}\}$ be the set of connected components before the update. Let us consider three basic cases (then it is straightforward to see that the proof holds also for combinations of these cases): (i) $C_i \in C$ is also a component of $C'$, (ii) $C_i \in C$ and $C_j \in C$ merge into one component $C'_k$ of $C'$, (iii) $C_i \in C$ splits into two components $C'_j$ and $C'_k$ of $C'$. In case (i), the \vd approximation of $C_i$ is updated exactly once in the for loop (Lines~\ref{vdup:loop1}~-~\ref{vdup:loop2}). In case (ii), (assuming $i<j$, w.l.o.g.) the \vd approximation of $C'_k$ is updated in the for loop from the source node $s_i \in C_i$. In its SSSP search, $s_i$ visits also $s_j \in C_j$, increasing $vis(s_j)$. Therefore, $s_j$ is skipped and exactly one \vd approximation is computed for $C'_k$. In case (iii), the source node $s_i \in C_i$ belongs to one of the components (say $C'_j$) after the update. During the for loop, the \vd approximation is computed for $C'_j$ via $s_i$. Also, for all the nodes $v$ in $C'_k$, $vis(v)$ is set to 0 and $v$ is inserted into $U$. Then some source node $s'_k \in C'_k$ must be extracted from $U$ in Line~\ref{vdup:extract} and a \vd approximation is computed for $C'_k$. Since all the nodes in $C'_k$ are set to visited during the search, no other \vd approximations are computed for $C'_k$.
 \qed
\end{proof}

\subsection{Proof of Theorem~\ref{thm:complexity_vd}}
\label{sub:proof_complexity_vd}
\begin{proof}
In the first part (Lines~\ref{vdup:loop1}~-~\ref{vdup:loop2} of Algorithm~\ref{algo:dyn_vd_approx}), we update an SSSP with \upvd or \upvdu for each source node $s_i$ such that $vis(s_i)$ is not greater than 1. Therefore the complexity of the first part is $O(n_c\cdot |\beta |\log |\beta |+\sum_{i=1}^{n_c}||A^{(i)}|| \log ||A^{(i)}||)$ in weighted graphs and $O(n_c\cdot |\beta |+\sum_{i=1}^{n_c}||A^{(i)}||+d^{(i)}_{max})$ in unweighted, for Lemmas~\ref{thm:complexity} and~\ref{thm:complexity2}. Only some of the affected nodes (those whose distance from a source node becomes equal to $\infty$) are inserted into the queue $U$. Therefore the cost of scanning $U$ in Lines~\ref{vdup:ext1}~-~\ref{vdup:ext2} is $O(\sum_{i=1}^{n_c}||A^{(i)}||)$. New SSSP searches are computed for new components that are not covered by the existing source nodes anymore. However, also such searches involve only the affected nodes and each affected node (and its incident edges) is scanned at most once during the search. Therefore, the total cost is $O(n_c\cdot |\beta |\log |\beta |+\sum_{i=1}^{n_c}||A^{(i)}|| \log ||A^{(i)}||)$ for weighted graphs and $O(n_c\cdot |\beta |+\sum_{i=1}^{n_c}||A^{(i)}||+d^{(i)}_{max})$ for unweighted graphs. \qed
\end{proof}

\subsection{Proof of Theorem~\ref{thm:correctness_bc}}
\label{sub:proof_correctness_bc}
\begin{proof}
Let $G$ be the old graph and $G'$ the modified graph after the batch of edge updates. Let $p'_{xy}$ be a shortest path of $G'$ between nodes $x$ and
$y$. To prove the theoretical guarantee, we need to prove that the probability of any sampled path $p'_{(i)}$ is equal to $p'_{xy}$ (\ie that the algorithms adds $1/r'$ to the nodes in $p'_{xy}$) is $\frac{1}{n(n-1)}\frac{1}{\sigma'_{x}(y)}$.
Algorithm~\ref{update} replaces the first $r$ shortest paths with other shortest paths $p'_{(1)},...,p'_{(r)}$ between the same node pairs (Lines~\ref{up:new1}~-~\ref{up:new2}) using Algorithm 4.1 of~\cite{DBLP:conf/alenex/BergaminiMS15}, for which it was already proven that $\Pr(p'_{(k)}=p'_{xy})=\frac{1}{n(n-1)}\frac{1}{\sigma'_{x}(y)}$~\cite[Theorem 4.1]{DBLP:conf/alenex/BergaminiMS15}. The additional $\Delta r$ shortest paths (Line~\ref{up:new_paths}) are recomputed from scratch with \rk, therefore also in this case $\Pr(p'_{(k)}=p'_{xy})=\frac{1}{n(n-1)}\frac{1}{\sigma'_{x}(y)}$ by Lemma 7 of~\cite{DBLP:conf/wsdm/RiondatoK14}.
 \qed
\end{proof}

\subsection{Proof of Theorem~\ref{thm:complexity_bc}}
\label{sub:proof_complexity_bc}
\begin{proof}
Let $\Delta r'$ be the difference between the values of $r'$ before and after the batch. Let us start from the simplest case: the graph $G$ is such that there is (before and after the update) one sample in each component and \vd does not increase after the update. This case includes, for example, connected graphs subject to a batch of only edge insertions, or any batch that neither splits the graph into more components nor increases \vd. In this case, $\Delta r=0$ and $\Delta r'=0$ and we only need to update the $r$ old shortest paths. Then, the total complexity is $O(r\cdot |\beta | + \sum_{i=1}^r (||A^{(i)}||+d_{\max}^{(i)}))$, where $A^{(i)}$ is the set of nodes affected in the $i$th SSSP, and $d_{\max}^{(i)}$ is the maximum distance in the $i$th SSSP. In general graphs, we might need to sample new paths for the betweenness approximation ($\Delta r>0$) and/or sample paths in new components that are not covered by any of the sampled paths ($\Delta r'>0$). Then, the complexity for the betweenness approximation update is $O(r\cdot |\beta | + \sum_{i=1}^r (||A^{(i)}||+d_{\max}^{(i)})) + O(\Delta r (|V|+|E|))$. The \vd update requires $O(r'\cdot |\beta | + \sum_{i=1}^{r'} (||A^{(i)}||+d_{\max}^{(i)}))$ to update the \vd approximation in the already covered components and $\sum_{i=1}^{\Delta r}(|V_i|+|E_i|)$ for the new ones, where $V_i$ and $E_i$ are nodes and edges of the $i$th component, respectively. \qed
\end{proof}

\section{Additional Experimental Results}
 \vspace{-4ex}
\label{sec:plots}
\begin{figure}[h]
 \vspace{-2ex}
\begin{center}
\includegraphics[width = 0.8\textwidth]{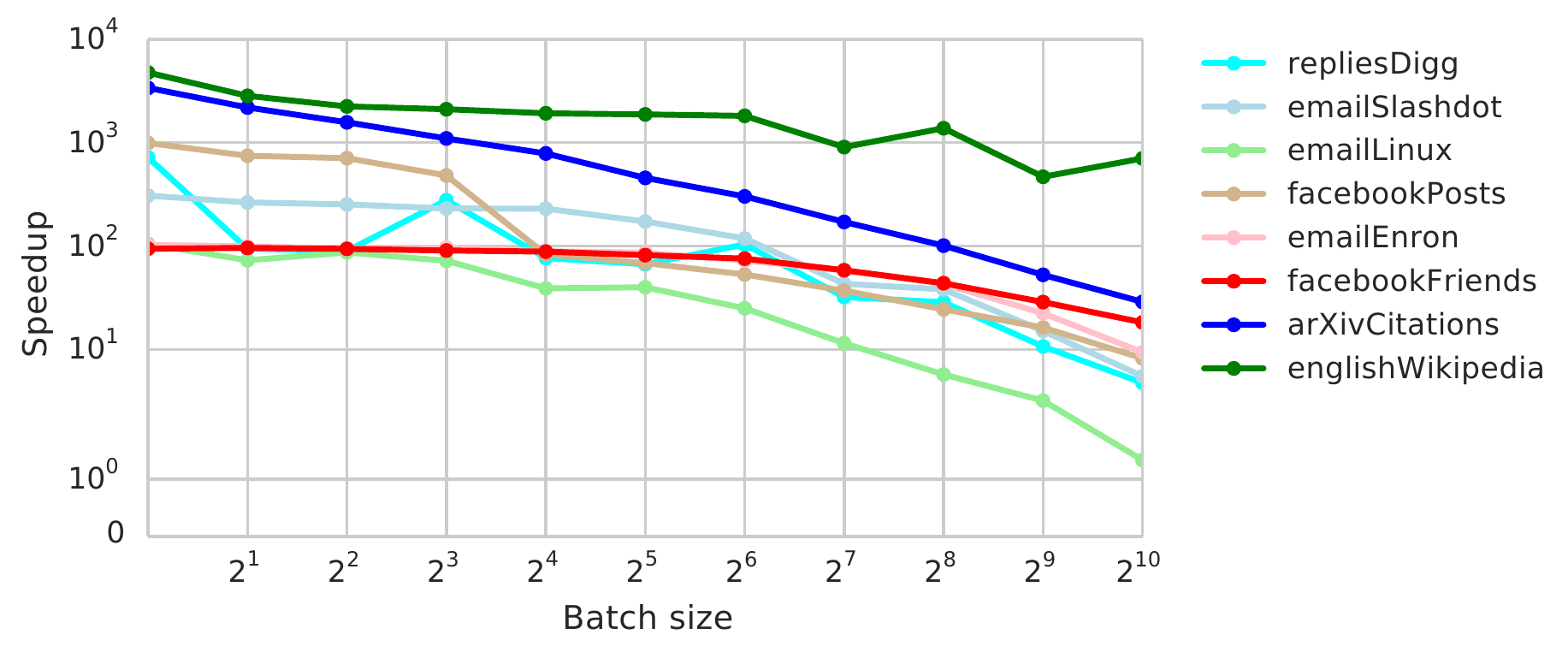}
\caption{Speedups on \rk in real unweighted graphs under random updates.}
\label{fig:random_unweighted}
\end{center}
 \vspace{-2ex}
\end{figure}
  \vspace{-2ex}
 \begin{table*}[h]
  \vspace{-4ex}
\begin{center}
\begin{scriptsize}
  \begin{tabular}{  l | r | r | r | r | r | r | r | r |}
\cline{2-9}  
 & \multicolumn{4}{  c |  }{Real}& \multicolumn{4}{  c |  }{Random} \\ \cline{2-9}
 &\multicolumn{2}{  c |  }{Time [s]}& \multicolumn{2}{  c |  }{Speedups} &\multicolumn{2}{  c |  }{Time [s]}& \multicolumn{2}{  c |  }{Speedups} \\ \cline{1-9}
\multicolumn{1}{| l|}{Graph} & $|\beta|=1$ & $|\beta|=1024$ & $|\beta|=1$ & $|\beta|=1024$  & $|\beta|=1$ & $|\beta|=1024$ & $|\beta|=1$ & $|\beta|=1024$ \\\cline{1-9}
\multicolumn{1}{| l|}{\texttt{repliesDigg}} 		& 0.053	& 3.032	& 605.18	& 14.24	& 0.049	& 3.046	& 658.19	& 14.17	\\ \cline{1-9}		
\multicolumn{1}{| l|}{\texttt{emailSlashdot}} 	& 0.790	& 5.387	& 50.81	& 16.12	& 0.716	& 5.866	& 56.00	& 14.81	\\ \cline{1-9}		
\multicolumn{1}{| l|}{\texttt{emailLinux}} 		& 0.324	& 24.816	& 5780.49	& 75.40	& 0.344	& 24.857	& 5454.10	& 75.28	\\ \cline{1-9}				
\multicolumn{1}{| l|}{\texttt{facebookPosts}} 	& 0.029	& 6.672	&  2863.83& 11.42	& 0.029	& 6.534	& 2910.33	& 11.66	\\ \cline{1-9}	
\multicolumn{1}{| l|}{\texttt{emailEnron}} 		& 0.050	& 9.926	& 3486.99	& 24.91	& 0.046	& 50.425	& 3762.09	& 4.90	\\ \cline{1-9}					
  \end{tabular}
  \end{scriptsize}
\end{center}
  \caption{Times and speedups of \daw on \rk in weighted real graphs under real dynamics and random updates, for batch sizes of 1 and 1024.}
  \label{table:weighted}
  \vspace{-4ex}
\end{table*} 
  \vspace{-2ex}
 \begin{table*}[h]
   \vspace{-4ex}
\begin{center}
\begin{scriptsize}
  \begin{tabular}{  l | r | r | r | r |}
\cline{2-5}  
 & \multicolumn{4}{  c |  }{Hyperbolic} \\ \cline{2-5}
 &\multicolumn{2}{  c |  }{Time [s]}& \multicolumn{2}{  c |  }{Speedups} \\ \cline{1-5}
\multicolumn{1}{| l|}{Number of edges} & $|\beta|=1$ & $|\beta|=1024$ & $|\beta|=1$ & $|\beta|=1024$ \\\cline{1-5}
\multicolumn{1}{| l|}{$m = 20000$} 				& 0.005	& 0.195	& 99.83	& 2.79	\\ \cline{1-5}		
\multicolumn{1}{| l|}{$m = 50000$} 				& 0.002	& 0.152	& 611.17	& 10.21	\\ \cline{1-5}
\multicolumn{1}{| l|}{$m = 200000$} 				& 0.015	& 0.288	& 422.81	& 22.64	\\ \cline{1-5}		
\multicolumn{1}{| l|}{$m = 500000$} 				& 0.012	& 0.339	& 1565.12	& 51.97	\\ \cline{1-5}		
\multicolumn{1}{| l|}{$m = 2000000$} 			& 0.049	& 0.498	& 2419.81	& 241.17	\\ \cline{1-5}		
\multicolumn{1}{| l|}{$m = 5000000$} 			& 0.083	& 0.660	& 4716.84	& 601.85	\\ \cline{1-5}			
\multicolumn{1}{| l|}{$m = 20000000$} 			& 0.006	& 0.401	& 304338.86	& 5296.78	\\ \cline{1-5}							
  \end{tabular}
  \end{scriptsize}
\end{center}
  \caption{Times and speedups of \da on \rk in hyperbolic unit-disk graphs, for batch sizes of 1 and 1024.}
  \label{table:hyperbolic}
\end{table*} 

\end{document}